\documentclass[12pt]{article}

\usepackage{graphicx}
\usepackage{amssymb}
\usepackage{amsthm}
\usepackage{geometry}
\usepackage{amsmath}
\usepackage{color}  % for color
\usepackage{pstricks}
\usepackage{pst-node}
\usepackage{pst-plot}
\usepackage{epstopdf}
\usepackage{caption}
\usepackage{algpseudocode} % Algorithms
\usepackage{algorithm}
\theoremstyle{plain}
\newtheorem{theorem}{Theorem}%[section]

\newtheorem{lemma}{Lemma}
\newtheorem{corollary}{Corollary}

\newcommand{\etal}{\emph{et al. }}

\newcommand{\ie}{{\it i.e.}, }   % \ie
\newcommand{\eg}{{\it e.g.}, }  %\eg
\newcommand{\cf}{{\it cf.}, }    % \cf
\newcommand{\np}{\mathsf{NP}}
\newcommand{\p}{\mathsf{P}}
\newcommand{\bone}{\mathbf{1}}  % bold "1"
\newcommand{\pow}{\mathcal{P}} % power set

\newcommand{\bs}[1]{\boldsymbol{#1}} % bold symbols in math mode, like \bf
\newcommand{\calU}{\mathcal{U}}
\newcommand{\calS}{\mathcal{S}}
\newcommand{\calC}{\mathcal{C}}
\newcommand{\bcalC}{\bs \calC} % bold \mathcal{C}
\newcommand{\bcalS}{\bs \calS} % bold \mathcal{S}
\def\mathbi#1{\textbf{\em #1}}
\newcommand{\A}{\mathbi{A}} % algorithm
\newcommand{\G}{\mathbi{G}} %  greedy algorithm

\algtext*{EndWhile}% Remove "end while" text
\algtext*{EndIf}% Remove "end if" text
\algtext*{EndFor}% Remove "end for" text

\title{Algorithm Instance Games}

\author{	Samuel D. Johnson \\
		 University of California, Davis \\
		 \texttt{samjohnson@ucdavis.edu} \\
		 \and
		 Tsai-Ching Lu \\
		 HRL Laboratories, LLC \\
		 \texttt{tlu@hrl.com}
		 }

\date{\today}
\begin{document}
\maketitle

%%%%%%%%%%%%%%%%%%%%%%%%%%%%%%%%%%%%%%%%
%% ABSTRACT
%%%%%%%%%%%%%%%%%%%%%%%%%%%%%%%%%%%%%%%%
\begin{abstract}
This paper introduces \emph{algorithm instance games (AIGs)} as a conceptual classification applying to games in which outcomes are resolved from joint strategies \emph{algorithmically}. For such games, a fundamental question asks: \emph{How do the details of the algorithm's description influence agents' strategic behavior?}

We analyze two versions of an AIG based on the \emph{set-cover} optimization problem. In these games, joint strategies correspond to instances of the set-cover problem, with each subset (of a given universe of elements) representing the strategy of a single agent. Outcomes are covers computed from the joint strategies by a set-cover algorithm. In one variant of this game, outcomes are computed by a deterministic greedy algorithm, and the other variant utilizes a non-deterministic form of the greedy algorithm. We characterize Nash equilibrium strategies for both versions of the game, finding that agents' strategies can vary considerably between the two settings. In particular, we find that the version of the game based on the deterministic algorithm only admits Nash equilibrium in which agents choose strategies (\ie subsets) containing at most one element, with no two agents picking the same element. On the other hand, in the version of the game based on the non-deterministic algorithm, Nash equilibrium strategies can include agents with zero, one, or every element, and the same element can appear in the strategies of multiple agents. 
\end{abstract}

%%%%%%%%%%%%%%%%%%%%%%%%%%%%%%%%%%%%%%%%
%% INTRODUCTION
%%%%%%%%%%%%%%%%%%%%%%%%%%%%%%%%%%%%%%%%
\section{Introduction}
\label{sec:intro}
Since its inception, the field of \emph{algorithmic game theory} has explored the consequences of placing algorithms in strategic settings. In this article, we take up a thread that can be traced through the tapestry of AGT research but -- to the best of our knowledge -- has not previously been explicitly identified. At its core, this thread involves settings in which outcomes are derived from joint strategies \emph{algorithmically}. In these settings, joint strategies may be viewed as instances that are passed to an algorithm tasked with returning outcomes, and the description of the game itself is tied to the (possibly implicit) use of a specific algorithm. To highlight the algorithm's role, we refer to these games as \emph{algorithm instance games (AIGs)}. 

Consider a \emph{hiring game} in which there is a single employer and $M$ candidates. Suppose that there is a set $\calU$ of all possible skills, and each candidate $i$ must invest in acquiring a subset $S_i \subset \calU$ of these skills. Skills are costly, with each one costing a fixed amount $\alpha > 0$. The candidates each submit a resume listing their skills, and the employer hires a subset of candidates, $C \subset M$, that has the broadest collective skill set but with the fewest number of individual candidates. If a candidate $i$ is hired, they receive a reward $\beta > 0$ less the cost spent on acquiring the skill set $S_i$; otherwise, candidate $i$ receives no reward but is still out the cost sunk into acquiring their skill set. From the perspective of the employer, the hiring game is clearly equivalent to the \emph{set-cover} optimization problem: presented with a collection of skill (sub)sets (\ie the candidates' resumes), select the fewest candidates who collectively cover the most skills. Because the hiring game's outcomes are computed algorithmically, and the realized outcome for a given joint strategy may be differ according to exactly which set-cover algorithm is used, a complete description of the hiring game must include a description of the algorithm used by the employer. 

For example, consider a hiring game in which there are three skills, $\calU = \{e_1, e_2, e_3\}$, and four candidates with skill sets $S_1 = \{e_1, e_2\}$, $S_2 = \{e_2, e_3\}$, $S_3 = \{e_3\}$, and $S_4 = \{e_1\}$. In this case, there are three equally optimal (from the perspective of the employer) covers\footnote{The three optimal covers are $\calC_a = \{S_1, S_2\}$, $\calC_b = \{S_1, S_3\}$, and $\calC_c = \{S_2, S_4\}$.} and which of these outcomes is realized will in general vary between one algorithm and another depending on the algorithm's internal decision making. More to the point, the algorithm's decisions will influence agents' strategy selections. Accordingly, a central research question motivating our research into AIGs asks: \emph{How does the implementation of the AIG's algorithm influence strategic decision making?}

In this paper, we examine the \emph{set-cover instance game (SCIG)}, an AIG based on the (unweighted) $\np$-complete set-cover problem (and equivalent to the hiring game example given above). We characterize Nash equilibrium strategies for two variants of the SCIG, each based on a version of the greedy set-cover approximation algorithm (for a performance analysis of the greedy algorithm, see \cite{Hochbaum1998} and the surveys \cite{Hoc1997-3,vaz2003,DoAA}). Both greedy algorithm iteratively builds a solution by selecting the subset with the most as yet uncovered elements, and is guaranteed to return a cover that is at most a logarithmic factor worse than the optimal (\ie minimal-cardinality) cover. The two versions of the greedy algorithm that we analyze correspond to \emph{deterministic} and \emph{non-deterministic} versions. In the deterministic version, the algorithm iterates following a fixed ordering of the subsets (that is known to the agents); and in the non-deterministic version, the algorithm iterates over the subsets following a uniform random permutation. Stated in terms of the hiring game, the two versions of the greedy algorithm can be thought of as follows: in the deterministic version, the employer always looks through the candidates' applications in alphabetical order, while in the non-deterministic version, the employer first shuffles the candidates' applications and then looks though it in which ever order happened to come out of this shuffling.

%%%%%%%%%%%%%%%%%%%%%%%%%%%%%%%%%%%%%%%%
\subsection{Discussion and Related Work}
\label{sec:intro:related}
Interpreted loosely, pretty much any game could be considered an AIG. After all, a giant lookup table that maps joint strategies to outcomes could be considered an algorithm. However, it is our intention that AIGs be construed more narrowly, reserving the designation for settings in which an understanding of the algorithmic aspects of the outcome resolution function that maps joint strategies to outcomes offers fundamental insights. For instance, one can define an AIG based on a combinatorial search or optimization problem and a corresponding search or optimization algorithm, just as we do in this paper with the SCIG. Another one could define an AIG based on an algorithm that simulates the decision making procedures of an individual or organization. Or, an AIG could be based on the implementation of a game theoretic mechanism that uses an  approximation algorithm to determine outcomes. Our motivation for conceiving the designation \emph{``algorithm instance game''} is to identify a natural, algorithmic perspective that has previously gone unidentified (in any general sense) in the algorithmic game theory literature, and we believe that this perspective will be of broader interest to the algorithmic game theory community. 

A number of existing mechanisms and games have characteristics that we believe qualify them as AIGs. For example, consider the \emph{winner determination problem (WDP)} for combinatorial auctions. The WDP is an optimization problem in which an algorithm (the mechanism) is given a bidding profile and must compute a feasible allocation of goods maximizing social welfare (\cf \cite{Cramton2006ch12}). The utility for an individual agent depends on the allocation they receive from the mechanism. In its full generality, the WDP problem presents a number of difficult computational challenges including the representation of bids \cite{Cramton2006ch9} and the computational complexity of finding the optimal allocation \cite{Rothkopf1998} (see also \cite{AGT2007ch11}). Special cases of the WDP corresponding to classic $\np$-hard optimization problems like integer linear programs and the knapsack problem have been studied algorithmically; for example, Bartal \etal \cite{Bartal2003} established inapproximability bounds on the social welfare of WDP mechanisms and Zurel and Nissan \cite{ZN2001} and Lehmann \etal \cite{Lehmann2006} analyze a WDP mechanism based on a greedy approximation algorithm. 

The \emph{PageRank game}, studied by Hopcroft and Sheldon \cite{Hopcroft2007,Sheldon2008}, can certainly be deemed an AIG, and their approach to analyzing the game's best-response strategies is consistent with the AIG perspective. In this network formation game, agents correspond to nodes and their strategies are sets of directed edges (from themselves to others). An agent's utility is proportional to their PageRank in the network comprised of all of the agents' directed links. In their equilibrium analysis, the authors find that an agent's best response is to only build links to those other agents whom link directly to them (\ie to reciprocate incoming links). Their result supports some well-known search engine optimization techniques (\eg link spam), and in so doing demonstrates an inadequacy of the PageRank algorithm as a reputation mechanism since it incentivizes agents to strategically build links that are inconsistent with the reputation system's objectives.

Finally, a number of covering games have appeared in the algorithm game theory literature in recent years\cite{Buchbinder2008,Cardinal2010,Escoffier2010,Piliouras2012,Balcan2013}. Although they are based on covering problems like set-cover, the work presented in these papers has little in common with AIGs or the SCIG presented in the current paper. Typically, agents in these covering games correspond to elements in the universal set, and the subsets are given as part of the specification of the game. An agent's strategy then corresponds to a selection of subsets containing their element, and the joint strategy induces a covering of the agents/elements.  The distinction is that, in these covering games, joint strategies correspond to \emph{solutions} to a covering problem, whereas in our work, the joint strategies correspond to \emph{instances} of a covering problem.

%%%%%%%%%%%%%%%%%%%%%%%%%%%%%%%%%%%%%%%%
\subsection{Summary of Results}
\label{sec:intro:results}

Our main analytical results are concerned with characterizing pure strategy Nash equilibrium for the two versions of the SCIG based on the greedy algorithm. We find that the two algorithms can result in Nash equilibrium that are considerably different.  In the SCIG based on the deterministic greedy algorithm, all Nash equilibrium strategies involve agents choosing at most one element, with no two agents choosing the same element. On the other hand, in the version of the SCIG based on the non-deterministic greedy algorithm, Nash equilibrium strategies can include agents with zero, one, or all $|\calU|$ elements. Moreover, if the equilibrium includes one agent choosing every element of $\calU$ in their strategy, there must be at least one other agent that also chooses every element of $\calU$ in their strategy. When the equilibrium strategy includes an agent with only one element, then there cannot simultaneously be another agent with more than one element in their strategy. Our main findings are summarized in Theorems~\ref{thm:d:ne} and \ref{thm:n:ne} for the deterministic and non-deterministic versions of the SCIG, respectively. 

The remainder of this paper is organized as follows: Section~\ref{sec:model} presents the SCIG model and Section~\ref{sec:analysis} presents our analysis of Nash equilibrium for the deterministic (\S~\ref{sec:analysis:d}) and non-deterministic (\S~\ref{sec:analysis:d}) versions of the game. The paper concludes with some final remarks in Section~\ref{sec:con}.

%%%%%%%%%%%%%%%%%%%%%%%%%%%%%%%%%%%%%%%%
%% MODEL
%%%%%%%%%%%%%%%%%%%%%%%%%%%%%%%%%%%%%%%%
\section{Model}
\label{sec:model}

The \emph{set-cover instance game (SCIG)} is an instance game that is derived from the \emph{set-cover} combinatorial optimization problem. Recall that an instance $\langle \calU, \calS \rangle$ of the set-cover problem consists of a universe of elements $\calU = \{e_1, \dots, e_n\}$ and a collection of subsets $\calS = \{S_1, \dots, S_m\}$ where each $S_i$ is a subset of $\calU$. The objective is to find a minimal cardinality collection $\calC \subseteq \calS$ that covers all of the elements in $\calU$; \ie $\calU = \bigcup_{S \in \calC} C$. 

In the SCIG, agents' strategies correspond to subsets of the universe of elements, $\calU$. Let $M = \{1, \dots, m\}$ denote the set of $m$ strategic agents. For each agent $i \in M$, $\bcalS_i = \pow(\calU)$ defines $i$'s strategy space,\footnote{$\pow(X)$ denotes the \emph{power-set} of $X$; \ie the set of all possible subsets of $X$.} where a pure strategy $S_i \in \bcalS_i$ is a (possibly empty) subset of $\calU$. Let $\bcalS = \bcalS_1 \times \cdots \times \bcalS_m$ denote the joint strategy space, with a joint strategy $\calS = (S_1, \dots, S_m) \in \bcalS$ representing a collection of $m$ subsets. A joint strategy $\calS$ and a (possibly truncated) universal set $\calU' = \bigcup_{i \in M} S_i$ designates an instance $\langle \calU', \calS \rangle$ of the set-cover problem. To simplify our notation, we will sometimes use $\calS$ to refer to both the joint strategy and the set-cover instance that it represents.

To finish the description of the SCIG, we need to specify a set-cover algorithm, $\A$. This algorithm is responsible for producing outcomes from joint strategy profiles. That is, given a joint strategy $\calS$, the algorithm returns a cover $\calC \leftarrow \A(\calU', \calS)$, and agents' utilities are then calculated with respect to this cover. Formally, given the joint strategy $\calS$, the utility for agent $i \in M$ is defined as
	\begin{equation}
	\label{eq:scig:utility}
	u_i(\calS) = \beta \bone(S_i \in \calC) - \alpha | S_i |,
	\end{equation}
where $\calC$ is the cover returned by the algorithm $\A$ given the set-cover instance $\calS$; $\beta > 0$ and $\alpha > 0$ are exogenously specified parameters (constants) that convey the benefit of being selected in the cover ($\beta$) and the cost of including elements in a strategy ($\alpha$); and $\bone(X)$ is an indicator function that equals one when condition $X$ is true and equals zero otherwise. 
								
If $\A$ is a non-deterministic algorithm, then it will not necessarily return a unique cover for a given instance. Let $\bcalC(\calS)$ be a random variable that takes as values the covers returned by $\A$ given the instance $\calS$, with $\Pr[\bcalC(\calS) = \calC]$ being the probability that a particular cover $\calC$ is returned. For a joint strategy $\calS$, the probability that an agent $i$'s strategy $S_i$ is part of a cover returned by the algorithm is 
	\begin{equation*}
	\Pr[S_i \in \bcalC(\calS)] = \sum_{\calC \in \bcalC(\calS)} \bone(S_i \in \calC) \Pr[\bcalC(\calS) = \calC].
	\end{equation*}
We define the (expected) utility for an agent $i \in M$ given a joint strategy $\calS$ and a non-deterministic algorithm $\A$ as
	\begin{equation}
	\label{eq:scig:utility2}
	u_i(\calS) = \beta \Pr[S_i \in \bcalC(\calS)] - \alpha | S_i |.
	\end{equation}

%%%%%%%%%%%%%%%%%%%%%%%%%%%%%%%%%%%%%%%%
\subsection{The Greedy Approximation Algorithm}
\label{sec:model:alg}

% Algorithm~\ref{alg:sc:greedy}
\begin{figure}[t]
\begin{algorithm}[H]
	\begin{algorithmic}[1]
		\Require $\calU = \{e_1, \dots, e_n\}$ is a universe of elements.
		\Require $\calS = \{S_1, \dots, S_m\}$ with $S_i \subseteq \calU$ for $i \in \{1, \dots, m\}$ is a collection of subsets.
		\Require $\calU = \bigcup_{S_i \in \calS} S_i$.
		\Function{$\G$}{$\calU, \calS$}
    			\State $\calC \gets \emptyset$
			\State $\calU' \gets \calU$
			\State $\pi \gets \text{permute}(1, \dots, m)$
			\While{$\calU' \neq \emptyset$}
				\State $S^* \gets S_{\pi(1)}$
				\For{$i = 2, 3, \dots, m$}
					\If{$| S^* \cap \calU' | < | S_{\pi(i)} \cap \calU' |$}
						\State $S^* \gets S_{\pi(i)}$
					\EndIf
				\EndFor
				\State $\calU' \gets \calU' \setminus S^*$
				\State $\calC \gets \calC \cup \{ S^* \}$
			\EndWhile 
			 \State \Return $\calC$
		\EndFunction
	\end{algorithmic}
	\caption{Greedy set-cover algorithm.}
	\label{alg:sc:greedy}
\end{algorithm}
\caption{In the deterministic version, $\G_d$, the permutation in line 4 is restricted to $\pi(i) = i$, for all $i = 1, \dots, m$. In the non-deterministic version, $\G_n$, the permutation $\pi$ is chosen uniformly at random from all permutations of $1, \dots, m$.}
\label{fig:sc:greedy}
\end{figure}

As an $\np$-complete problem, there is no known efficient set-cover algorithm that is guaranteed to always return an optimal cover $\calC$ \cite{GJ1979}. Here, we analyze two variants of the simple \emph{greedy} approximation algorithm, $\G$ (see Figure~\ref{fig:sc:greedy}): a \emph{deterministic} version that always iterates through the collection of subsets $\calS = (S_1, \dots, S_m)$ in a fixed order that is known  by every agent; and a \emph{non-deterministic} version that iterates through $\calS$ following a permutation $\pi$ that is selected uniformly at random from all permutations of $1, \dots, m$. Through the remainder of this paper, we will let $\G_d$ denote the deterministic greedy algorithm, and $\G_n$ denote the non-deterministic greedy algorithm. Also, we will use SCIG($\G_d$) and SCIG($\G_n$) to denote the set-cover instance games based on the algorithms $\G_d$ and $\G_n$, respectively. 

We note that the greedy algorithm is guaranteed to return a cover that is at most a $O(\log n)$ factor worse than optimal, which (assuming $\p \neq \np$) is essentially the best possible approximation one can get from a polynomial-time algorithm for the set-cover problem \cite{Hoc1997-3,vaz2003,DoAA}.

%%%%%%%%%%%%%%%%%%%%%%%%%%%%%%%%%%%%%%%%
%% ANALYSIS
%%%%%%%%%%%%%%%%%%%%%%%%%%%%%%%%%%%%%%%%
\section{Equilibrium Analysis}
\label{sec:analysis}

In this section we characterize the pure Nash equilibrium outcomes for the set-cover instance game. Recall that a joint strategy $\calS = (S_1, \dots, S_m)$ is a pure strategy Nash equilibrium if, for every agent $i \in M$ and every strategy $S_i' \in \bcalS_i$,
	\begin{equation*}
	u_i(\calS) \geq u_i(S_i', \calS_{-i})
	\end{equation*}
where $\calS_{-i}$ denotes the joint strategy of every agent $j \neq i$.

%%%%%%%%%%%%%%%%%%%%%%%%%%%%%%%%%%%%%%%%
\subsection{The Deterministic Case}
\label{sec:analysis:d}

Pure strategy Nash equilibrium for the SCIG based on the deterministic greedy algorithm $\G_d$ consist exclusively of agents with zero or one-element strategies -- see Theorem~\ref{thm:d:ne} below. The basic intuition into why strategies with more that a single element are not stable follows from observing two facts. First, as we show in Lemma~\ref{lem:d:no_overlap} below, no two agents $i$ and $j$ will share any common elements in a Nash equilibrium. Because $\G_d$ is deterministic, if the strategies of both $i$ and $j$ get selected then it will be the case that one of them is always selected before the other. Therefore, the agent that is selected second can safely drop those common elements and still ensure that their strategy is part of the cover returned by $\G_d$. The second fact is due to the definition of a Nash equilibrium -- that an agent must be playing their best-response strategy given all other agents' strategies remaining unchanged. Putting these two facts together, the best-response for an agent $i$ whose strategy $S_i$ is selected by $\G_d$ is to switch to the strategy $S_i' = \{e_k\}$ for some $e_k \in S_i$ since $e_k$ is not going to be part of any other agent's strategy in Nash equilibrium and $S_i'$ will still get chosen by $\G_d$ but at a lower cost.

% Theorem~\ref{thm:d:ne}
\begin{theorem}
\label{thm:d:ne}
In the SCIG($\G_d$), every Nash equilibrium strategy $\calS$ has agents with non-overlapping strategies. Moreover,
	\begin{enumerate}
	\item \label{thm:d:ne:1} 
	If $\alpha > \beta$, then the only Nash equilibrium has $S_i = \emptyset$ for all $i \in M$.
	\item \label{thm:d:ne:2} 
	If $\beta / 2 \leq \alpha < \beta$, then Nash equilibrium strategies have the first $\ell_1 = \min \{m, n\}$ agents with distinct, single-element strategies and all remaining $\ell_0 = \max \{0, m - \ell_1 \}$ agents have empty strategies.
	\item \label{thm:d:ne:3}
	If $\alpha = \beta$, then a Nash equilibrium strategy will involve up to $n$ agents choosing distinct, single-element strategies and all other agents choosing empty strategies.
	\end{enumerate}
\end{theorem}

The proof of Theorem~\ref{thm:d:ne} is presented below. First, we establish the a couple of lemmata. Lemma~\ref{lem:d:no_overlap} establishes that in a Nash equilibrium, no two agents' strategies share common elements; and Lemma~\ref{lem:d:cycle} bounds the maximum size of an agent $i$'s strategy $S_i$ in a Nash equilibrium to $|S_i| \leq 1$.

% Lemma~\ref{lem:d:no_overlap}
\begin{lemma}
\label{lem:d:no_overlap}
A joint strategy $\calS$ in which there exist a pair of (distinct) agents $i, j \in M$ with $S_i \cap S_j \neq \emptyset$ cannot be a Nash equilibrium. 
\end{lemma}

\begin{proof} % proof of Lemma~\ref{lem:d:no_overlap}
If either $S_i$ is not selected in the cover returned by $\G_d$ then $i$ can strictly improve its utility by choosing the strategy $S_i = \emptyset$, which would clearly share no common elements with $S_j$. A symmetric argument holds for $j$ if $S_j$ is not selected by $\G_d$. 

Suppose now that both $S_i$ and $S_j$ are in the cover returned by $\G_d$, and that $\G_d$ selects $S_i$ before selecting $S_j$. Assume, without loss of generality, that $S_i \cap S_j = \{e_k\}$ and no other subset contains $e_k$. Since $\G_d$ selects $S_i$ before $S_j$ (and $\G_d$ is picking $S_j$ on the merit of $S_j$'s uncovered elements), then $j$ can deviate to a strategy $S_j' = S_j \setminus \{e_k\}$ and still be selected by $\G_d$. Moreover, such a deviation would increase $j$'s utility by $\alpha$.
\end{proof} % end of proof of Lemma~\ref{lem:d:no_overlap}

% Lemma~\ref{lem:d:cycle}
\begin{lemma}
\label{lem:d:cycle}
There does not exist a Nash equilibrium strategy $\calS$ in which there is an agent $i \in M$ with $|S_i| > 1$.
\end{lemma}

\begin{proof} % proof of Lemma~\ref{lem:d:cycle}
By Lemma~\ref{lem:d:no_overlap}, the elements in $S_i$ are not in any other agent's strategy. Therefore, $i$ could select one element $e_k \in S_i$ and deviate to the strategy $S_i' = \{e_k\}$ and still be selected by $\G_d$ and increase their utility by $( |S_i| - 1 ) \alpha$.
\end{proof} % end of proof ofLemma~\ref{lem:d:cycle}

Given Propositions~\ref{lem:d:no_overlap} and \ref{lem:d:cycle}, we can now prove Theorem~\ref{thm:d:ne}.

\begin{proof}[Proof of Theorem~\ref{thm:d:ne}]
Lemma~\ref{lem:d:no_overlap} establishes that agents' strategies do not share any common elements in Nash equilibrium. Part~\ref{thm:d:ne:1} follows from the observation that whenever $\alpha > \beta$ then the cost to an agent of adding even a single element to their strategy is greater than the benefit they would get for being selected in a cover. 

Similarly, part~\ref{thm:d:ne:3} addresses the case that the cost of a single element is equal to the benefit that comes from being selected in a cover. It is easy to see that in either case the best that an agent can hope to do is receive a utility of zero.

Finally, part~\ref{thm:d:ne:2} follows from Lemma~\ref{lem:d:cycle} and the observation that when $\alpha < \beta$, a strategy$S_i = \{e_k\}$ with a distinct, single element $e_k$ that does not belong to any other strategy (\ie $e_k \notin S_j$ for all $j \neq i$) yields a positive utility for agent $i$.
\end{proof} % end of proof of Theorem~\ref{thm:d:ne}

We wrap up the analysis of the SCIG($\G_d$) with the following theorem, establishing the non-existence of pure strategy Nash equilibrium for a significant part of the SCIG's parameter space. The intuition behind its proof involves a coupling of the \emph{pigeonhole principle} (\ie there are more agents then elements) with the upper bound on strategy sizes given by Lemma~\ref{lem:d:cycle}.

% Theorem~\ref{thm:d:br}
\begin{theorem}
\label{thm:d:br}
There does not exist a (pure strategy) Nash equilibrium whenever $m > n$ and $\alpha < \beta / 2$.
\end{theorem}

\begin{proof} % proof of Theorem~\ref{thm:d:br}
Suppose, toward a contradiction, that $\calS$ is a Nash equilibrium. From Propositions~\ref{lem:d:no_overlap} and \ref{lem:d:cycle}, we know that if $\calS$ includes an agent $i$ with $S_i \neq \emptyset$, then it must be that $|S_i| = 1$ and the element in $S_i$ cannot appear in any other $S_j$. 

Since $\alpha < \beta / 2$, an agent with a single-element strategy receives a higher utility than an agent with an empty strategy. Therefore, the joint strategy $\calS$ must include $n = |\calU|$ agents with distinct, single-element strategies, with the remaining $m - n = |M| - |\calU|$ agents having empty strategies. However, such a strategy cannot be a Nash equilibrium since an agent $i$ with $S_i = \emptyset$ can change to a strategy $S_i' = S_j \cup S_k$ for some pair of agents $j, k \in M$ with $| S_j | = | S_k | = 1$, which will ensure their selection by $\G_d$, and earn $i$ a higher utility:
	$$u_i(S_i', \calS_{-i}) = \beta - 2 \alpha > 0 = u_i(S_i, \calS_{-i}).$$
\end{proof} % end of proof of Theorem~\ref{thm:d:br}

%%%%%%%%%%%%%%%%%%%%%%%%%%%%%%%%%%%%%%%%
\subsection{The Non-Deterministic Case}
\label{sec:analysis:nd}

In the SCIG based on the non-deterministic greedy algorithm $\G_n$, the range of possible Nash equilibrium strategies is a bit richer than those based on the version of the game using the deterministic greedy algorithm. These outcomes are described in Theorem~\ref{thm:n:ne}. 

Our main insight into the SCIG($\G_n$) is captured in Lemma~\ref{lem:n:growing_strat}, which identifies an interesting pressure driving best-response strategy selections and ruling out the existence of Nash equilibrium for a large part of the game's parameter space. The basic idea is as follows: Suppose that there are more agents than elements, and there is an agent $i$ with a strategy $S_i$ such that, for all other agents $j \neq i$, $|S_j| \leq |S_i| < | \calU |$. If the probability that $\G_n$ selects $S_i$ for a cover is less than one, then $i$ could switch to a strategy $S_i'$ that includes all of the elements in $S_i$ and one additional element to increase their probability of being selected by $\G_n$ to one. But now that $S_i'$ is always selected by $\G_n$, any agent $j$ for whom $S_j \cap S_i' \neq \emptyset$ will be better off by switching to a strategy $S_j' \ S_j \setminus S_i'$. This now allows $i$ to deviate to a third strategy $S_i'' = \{e_k\}$ consisting of only a single element $e_k \in S_i'$. Since $e_k$ was unique to $S_i'$ after every other agent $j$ dropped from their strategies elements in $S_i''$, it is therefore still unique to $i$'s strategy after switching to $S_i''$, which ensures that the probability that $S_i''$ is in a cover returned by $\G_n$ remains one. This behavior is similar to what was observed in the SCIG($\G_d$) model, which prevented agents from having strategies with more than a single element in a Nash equilibrium. However, in the SCIG($\G_d$) model, there is the possibility of agents having more than a single element in a Nash equilibrium. But in order for this to happen, at least two agents $i$ and $j$ would need to choose strategies $S_i = S_j = \calU$; see Corollary~\ref{cor:n:full_empty}.

% Theorem~\ref{thm:n:ne}
\begin{theorem}
\label{thm:n:ne}
In the SCIG($\G_n$), pure Nash equilibrium strategies are:
	\begin{enumerate}
	\item \label{thm:n:ne:1} 
	If $m \leq n$ and $\alpha \leq \beta$, then Nash equilibrium strategies will involve disjoint, single-element subsets and/or empty subsets.
	\item \label{thm:n:ne:2} 
	If $m > n$ and $\alpha = \beta / 2$, then there are Nash equilibrium strategies in which either:
		\begin{itemize}
		\item All agents have zero- or single-element strategies, and for every agent $i$ with a single-element strategy $S_i = \{e_k\}$, there exists exactly one other agent $j$ with $S_j = \{e_k\}$; or
		\item All agents have single-element strategies with each element appearing in either one or two individual strategies.
		\end{itemize}
	\item \label{thm:n:ne:3} 
	For positive integers $\rho \geq 2$ satisfying $m > \rho n$ and $\beta / (1 + \rho) n < \alpha \leq \beta / \rho n$, then Nash equilibrium strategies will have $\rho$ (or possibly $\rho - 1$ when $\rho \geq 3$) agents $i$ with $S_i = \calU$ and the remaining $m - \rho$ (or possibly $m - \rho + 1$ when $\rho \geq 3$) agents $j$ with $S_j = \emptyset$.
	\end{enumerate}
\end{theorem}

% Lemma~\ref{lem:n:growing_strat}
\begin{lemma}
\label{lem:n:growing_strat}
When $\alpha < \beta / 2$ and $m > n$, if $\calS$ is a Nash equilibrium and there exists an agent $i$ for whom $0 < \Pr[S_i \in \bcalC(\calS)] < 1$ then it must be the case that $S_i = \calU$.
\end{lemma}

\begin{proof} % proof of Lemma~\ref{lem:n:growing_strat}
Assume that $\calS$ is a Nash equilibrium. Set $\mu = \max_{j \in M} |S_j|$, and denote by $\calS_{max} \subset \calS$ the set of strategies in $\calS$ with cardinalities $\mu$; that is, $\calS_{max} = \{ S_j : |S_j| = \mu \}$. By definition, the greedy algorithm $\G_n$ will always select an element from $\calS_{max}$ in its first iteration. 

Suppose that $1 < \mu < n$ and let $S_i \in \calS_{max}$ be one of these $\mu$-cardinality strategies. Since $|S_i| > 0$, it must be the case that $u_i(\calS) \geq 0$ (otherwise $\calS$ would certainly not be a Nash equilibrium). From our assumption that $\Pr[S_i \in \bcalC(\calS)] < 1$, we can infer that agent $i$'s utility is at most $u_i(\calS) \leq \beta / 2 - \alpha \mu$. However, by deviating to a strategy $S_i' = S_i \cup \{e_k\}$ for some arbitrary element $e_k \in \calU \setminus S_i$, agent $i$ can weakly increase their utility to $u_i(S_i', \calS_{-i}) = \beta - (1 + \mu) \alpha \geq u_i(\calS)$, contradicting our assumption that $\calS$ was a Nash equilibrium.

Suppose instead that $\mu = n$ and let $S_i \in \calS \setminus \calS_{max}$ be a strategy with cardinality less than $\mu$. In this case, $\G_n$ will never pick $S_i$ since the algorithm will halt after picking an element from $\calS_{max}$ in its first iteration. This contradicts our assumption that $\Pr[S_i \in \bcalC(\calS)] > 0$.
\end{proof} % end of proof of Lemma~\ref{lem:n:growing_strat}

The next corollary follows immediately from Lemma~\ref{lem:n:growing_strat}.

% Corollary~\ref{cor:n:full_empty}
\begin{corollary}
\label{cor:n:full_empty}
If $\calS$ is a Nash equilibrium in which there is at least one agent $i$ with $S_i = \calU$, then for every other agent $j$ either $S_j = \calU$ or $S_j = \emptyset$.
\end{corollary}

We now turn to proving the main result of this section -- Theorem~\ref{thm:n:ne}.

\begin{proof}[Proof of Theorem~\ref{thm:n:ne}]
We start by proving part~\ref{thm:n:ne:1}. If $\alpha < \beta$, then all $m$ agents will choose a single element in their strategy, and since $n \geq m$, no two agents will choose the same element. If $\alpha = \beta$, then agents will be indifferent between choosing a single (distinct) element or an empty strategy since both yield the same utility (zero).

For part~\ref{thm:n:ne:2}, if $\alpha = \beta / 2$ then we can create a Nash equilibrium as follows: Denoting the agents by $M = \{1, 2, \dots, m\}$ and the universal set of elements by $\calU = \{e_1, e_2, \dots, e_n\}$, set $S_i = S_{i + n} = \{e_i\}$ for agents $i \in \{ 1, \dots, m \}$. If $i+n > m$ then just set $S_i = \{e_i\}$, and if $m > 2n$ then, for the remaining agents $j \in \{ 2n, 2n+1, \dots, m\}$, set $S_j = \emptyset$. It is easy to verify that this results in a Nash equilibrium since no agent would stand to (strictly) gain by dropping their element (if they have one), nor would they (strictly) gain by adding an additional element.

For part~\ref{thm:n:ne:3}, fix a value for $\rho \geq 3$. If $\calS$ is a Nash equilibrium with exactly $\rho$ agents $i$ with $S_i = \calU$, then $\G_n$ selects a particular one of these subsets with probability $1 / \rho$. In order for $S_i = \calU$ to be a best response given that at most $\rho - 1$ other agents $j$ have $S_j = \calU$, it must be the case that $\alpha \leq \beta / \rho n$. (If there are exactly $\rho - 1$ other agents $j$ with $S_j = \calU$ and $\alpha = \beta / \rho n$ then agent $i$'s best response could be either $S_i = \calU$ or $S_i = \emptyset$.)

If $\rho = 2$, then the only Nash equilibrium for this regime is when there are two agents $i$ and $j$ with $S_i = S_j = \calU$. If instead there was only one agent $i$ with $S_i = \calU$ then they could deviate to the strategy $S_i' = \{e_k\}$ for some arbitrary element $e_k \in \calU$ and strictly increase their utility. (A similar argument holds for why there is no Nash equilibrium for the case when $\rho = 1$.)

Finally, from the preceding argument, Lemma~\ref{lem:n:growing_strat}, and Corollary~\ref{cor:n:full_empty}, it follows that, if $m > 2n$ and $\alpha < \beta / 2$, any Nash equilibrium must include at least two agents $i$ and $j$ with $S_i = S_j = \calU$.
\end{proof} % end proof of Theorem~\ref{thm:n:ne}

%%%%%%%%%%%%%%%%%%%%%%%%%%%%%%%%%%%%%%%%
%% CONCLUDING REMARKS
%%%%%%%%%%%%%%%%%%%%%%%%%%%%%%%%%%%%%%%%
\section{Concluding Remarks}
\label{sec:con}

This paper identified the notion of \emph{algorithm instance games (AIGs)}, referring to game theoretic models (\eg games and mechanisms) in which outcomes are resolved from joint strategy profiles algorithmically.  AIGs offer a perspective into a fundamentally algorithmic aspect of game theory that has heretofore received only incidental attention in isolated contexts. The primary conceptual contribution of this paper is the explicit identification of AIGs, which we believe will appeal to the broader interests of the algorithmic game theory community, presenting an interesting topic for future research.

Our primary technical contribution involved the characterization of pure Nash equilibrium strategies for two versions of the \emph{set-cover instance game (SCIG)} -- an AIG based on the set-cover optimization problem. Both versions of this game utilize the simple greedy algorithm for resolving outcomes from agents' joint strategy profiles. They differ in that in one version, which we denote by SCIG($\G_d$), the greedy algorithm iterates according to a fixed ordering of the agents' strategies, and in the other, denoted SCIG($\G_n$), the greedy algorithm iterates according to a permutation of agents' strategies chosen uniformly at random. Our analysis of these games demonstrate how a slight change in the outcome resolution algorithm can lead to significantly different equilibrium predictions. Or, cast another way, our findings show that agents' knowledge of the greedy algorithm's iteration sequence \emph{a priori} will have a material effect on the set of Nash equilibrium strategies.

%%%%%%%%%%%%%%%%%%%%%%%%%%%%%%%%%%%%%%%%
\subsection{Acknowledgments}
\label{sec:con:ack}
We would like to extend our thanks to our colleagues at UC Davis and the anonymous reviewers who's feedback on an earlier version of this paper helped to improve it considerably.

%%%%%%%%%%%%%%%%%%%%%%%%%%%%%%%%%%%%%%%%
%% BIBLIOGRAPHY
%%%%%%%%%%%%%%%%%%%%%%%%%%%%%%%%%%%%%%%%
\bibliographystyle{alpha}
\bibliography{bib}

\begin{thebibliography}{BLENO08}

\bibitem[BBM13]{Balcan2013}
Maria-Florina Balcan, Avrim Blum, and Yishay Mansour.
\newblock The price of uncertainty.
\newblock {\em ACM Trans. Econ. Comput.}, 1(3):15:1--15:29, September 2013.

\bibitem[BGN03]{Bartal2003}
Yair Bartal, Rica Gonen, and Noam Nisan.
\newblock Incentive compatible multi unit combinatorial auctions.
\newblock In {\em Proceedings of the 9th Conference on Theoretical Aspects of
  Rationality and Knowledge}, TARK '03, pages 72--87, New York, NY, USA, 2003.
  ACM.

\bibitem[BLENO08]{Buchbinder2008}
Niv Buchbinder, Liane Lewin-Eytan, Joseph(Seffi) Naor, and Ariel Orda.
\newblock Non-cooperative cost sharing games via subsidies.
\newblock In Burkhard Monien and Ulf-Peter Schroeder, editors, {\em Algorithmic
  Game Theory}, volume 4997 of {\em Lecture Notes in Computer Science}, pages
  337--349. Springer Berlin Heidelberg, 2008.

\bibitem[BN07]{AGT2007ch11}
Liad Blumrosen and Noam Nisan.
\newblock Combinatorial auctions.
\newblock In Noam Nisan, Tim Roughgarden, \'{E}va Tardos, and Vijay~V.
  Vazirani, editors, {\em Algorithmic Game Theory}, chapter~11. Cambridge
  University Press, New York, NY, USA, 2007.

\bibitem[CH10]{Cardinal2010}
Jean Cardinal and Martin Hoefer.
\newblock Non-cooperative facility location and covering games.
\newblock {\em Theoretical Computer Science}, 411(16-18):1855 -- 1876, 2010.

\bibitem[EGM10]{Escoffier2010}
Bruno Escoffier, Laurent Gourv{\`e}s, and J{\'e}r{\^o}me Monnot.
\newblock On the impact of local taxes in a set cover game.
\newblock In Boaz Patt-Shamir and Tınaz Ekim, editors, {\em Structural
  Information and Communication Complexity}, volume 6058 of {\em Lecture Notes
  in Computer Science}, pages 2--13. Springer Berlin Heidelberg, 2010.

\bibitem[GJ79]{GJ1979}
Michael~R. Garey and David~S. Johnson.
\newblock {\em Computers and Intractability: A Guide to the Theory of
  $\mathbf{NP}$-Completeness}.
\newblock W. H. Freeman and Company, New York, NY, USA, 1979.

\bibitem[Hoc97]{Hoc1997-3}
Dorit~S. Hochbaum.
\newblock Approximating covering and packing problems: Set cover, vertex cover,
  independent set, and related problems.
\newblock In Dorit~S. Hochbaum, editor, {\em Approximation Algorithms for
  {NP}-Hard Problems}, chapter~3, pages 94--143. {PWS} Publishing Company,
  Boston, MA, USA, 1997.

\bibitem[HP98]{Hochbaum1998}
Dorit~S. Hochbaum and Anu Pathria.
\newblock Analysis of the greedy approach in problems of maximum k-coverage.
\newblock {\em Naval Research Logistics (NRL)}, 45(6):615--627, 1998.

\bibitem[HS07]{Hopcroft2007}
John Hopcroft and Daniel Sheldon.
\newblock Manipulation-resistant reputations using hitting time.
\newblock In {\em Algorithms and Models for the Web-Graph: 5th International
  Workshop, WAW 2007, San Diego, CA, USA, December 11-12, 2007. Proceedings},
  volume 4863 of {\em Lecture Notes in Computer Science}, pages 68--81, Berlin,
  Heidelberg, 2007. Springer-Verlag.

\bibitem[HS08]{Sheldon2008}
John Hopcroft and Daniel Sheldon.
\newblock Network reputation games.
\newblock Technical report, Cornell University, Ithaca, NY 14850, October 2008.

\bibitem[LLN06]{Lehmann2006}
Benny Lehmann, Daniel Lehmann, and Noam Nisan.
\newblock Combinatorial auctions with decreasing marginal utilities.
\newblock {\em Games and Economic Behavior}, 55(2):270--296, May 2006.

\bibitem[LMS06]{Cramton2006ch12}
Daniel Lehmann, Rudolf M\"{u}ller, and Tuomas Sandholm.
\newblock The winner determination problem.
\newblock In Peter Cramton, Yoav Shoham, and Richard Steinberg, editors, {\em
  Combinatorial Auctions}, chapter~12. MIT Press, Cambridge, MA, 2006.

\bibitem[Nis06]{Cramton2006ch9}
Noam Nisan.
\newblock Bidding languages for combinatorial auctions.
\newblock In Peter Cramton, Yoav Shoham, and Richard Steinberg, editors, {\em
  Combinatorial Auctions}, chapter~9. MIT Press, Cambridge, MA, 2006.

\bibitem[PVV12]{Piliouras2012}
Georgios Piliouras, Tom\'{a}\v{s} Valla, and L\'{a}szl\'{o}~A. V\'{e}gh.
\newblock Lp-based covering games with low price of anarchy.
\newblock In Paul~W. Goldberg, editor, {\em Internet and Network Economics},
  volume 7695 of {\em Lecture Notes in Computer Science}, pages 184--197.
  Springer Berlin Heidelberg, 2012.

\bibitem[RPH98]{Rothkopf1998}
Michael~H. Rothkopf, Aleksandar Peke{\v c}, and Ronald~M. Harstad.
\newblock Computationally manageable combinational auctions.
\newblock {\em Management Science}, 44(8):1131--1147, 1998.

\bibitem[Vaz01]{vaz2003}
Vijay Vazirani.
\newblock {\em Approximation Algorithms}.
\newblock Springer, Berlin, 2001.

\bibitem[WS11]{DoAA}
David~P. Williamson and David~B. Shmoys.
\newblock {\em The Design of Approximation Algorithms}.
\newblock Cambridge University Press, Cambridge, April 2011.

\bibitem[ZN01]{ZN2001}
Edo Zurel and Noam Nisan.
\newblock An efficient approximate allocation algorithm for combinatorial
  auctions.
\newblock In {\em EC '01: Proceedings of the 3rd ACM conference on Electronic
  Commerce}, pages 125--136, New York, NY, USA, 2001. ACM.

\end{thebibliography}

\end{document}